\documentclass{lmcs}
\pdfoutput=1

\usepackage{lastpage}
\lmcsdoi{16}{3}{7}
\lmcsheading{}{\pageref{LastPage}}{}{}%
{Oct.~05,~2018}{Jul.~23,~2020}{}

\usepackage{graphicx}

\usepackage[utf8]{inputenc}

\keywords{Fixed point combinator, Lambda Calculus, Bohm tree, FPC generator}

\usepackage{MnSymbol}

\usepackage{bussproofs}
\EnableBpAbbreviations

\newcommand{\wFPC}{\mathsf{(W)FPC}}

\usepackage{amsmath}
\usepackage{stmaryrd}     
\usepackage{MnSymbol}   
\usepackage{bussproofs}  
\EnableBpAbbreviations

\DeclareFontFamily{OT1}{pzc}{}   
\DeclareFontShape{OT1}{pzc}{m}{it}{<-> s * [1.10] pzcmi7t}{}
\DeclareMathAlphabet{\mathpzc}{OT1}{pzc}{m}{it}

\newcommand{\lam}{\lambda}
\newcommand{\Lam}{\Lambda}

\newcommand{\ic}{\mathtt{I}}
\newcommand{\kc}{\mathtt{K}}
\newcommand{\cnc}{\mathtt{c}}       
\newcommand{\yc}{\mathtt{Y}}

\newcommand{\setof}[1]{\{{#1}\}}

\newcommand{\fv}{\mathsf{FV}}

\newcommand{\comment}[1]{}

\newcommand{\RA}{\Rightarrow}
\newcommand{\LA}{\Leftarrow}

\newcommand{\then}{\;\Longrightarrow\;}

\newcommand{\thra}{\twoheadrightarrow}
\newcommand{\thla}{\twoheadleftarrow}

\newcommand{\df}{\quad:=\quad}

 \newcommand{\Ups}{\Upsilon}
 
 \newcommand{\notinfty}{\notin^\infty}
 \newcommand{\bte}{\mathbin{=^\infty}}

\begin{document}

\title[Fixed point combinators as fixed points of FPC generators]{Fixed point combinators as fixed points
of higher-order fixed point generators}
\author{Andrew Polonsky}
\address{Appalachian State University\\
Boone NC 28608-2133, USA}
\email{andrew.polonsky@gmail.com}

\begin{abstract}
Corrado B\"ohm once observed that if $Y$ is any fixed point combinator (fpc), then $Y(\lam yx.x(yx))$ is again fpc.  He thus discovered the first ``fpc generating scheme'' -- a generic way to build new fpcs from old.  Continuing this idea, define an \emph{fpc generator} to be any sequence of terms $G_1,\dots,G_n$ such that
\[       Y \text{ is fpc } \then  YG_1\cdots G_n \text{ is fpc}.\]
In this contribution, we take first steps in studying the structure of (weak) fpc generators.
We isolate several robust classes of such generators, by examining their elementary properties like injectivity and (weak) constancy.
We provide sufficient conditions for existence of fixed points of a given generator $(G_1,\cdots,G_n)$: an fpc $Y$ such that $Y = Y G_1 \cdots G_n$.
We conjecture that weak constancy is a necessary condition for existence of such (higher-order) fixed points.  This statement generalizes Statman's conjecture on non-existence of ``double fpcs'': fixed points of the generator
$(G) = (\lam yx.x(yx))$ discovered by B\"ohm.

Finally, we define and make a few observations about the monoid of (weak) fpc generators.
This enables us to formulate new conjectures about their structure.
\end{abstract}

\maketitle

{\centering \large \textsl{Dedicated to Corrado B\"ohm, a pioneer of the Lambda Calculus} \par}

\section{Introduction}

Fixed point combinators (fpcs) are a fascinating class of lambda terms.  Arising in the proof of the Fixed Point Theorem, their dynamical character affects the global structure of the Lambda Calculus in a fundamental way.  Being a mechanism of unrestricted recursion,
they are directly responsible for the Turing-completeness of the lambda calculus as a programming language.\footnote{
In fact, the very notion of Turing-completeness
traces back to Church's bold suggestion that lambda calculus can encode arbitrary computational processes.  Yet the idea was only accepted after Kleene and Turing, using fixed
point constructions, showed equivalence between Church's formalism
and their own ones.}
And when lambda terms are used as the computational basis of a logical system\\
\phantom{.}{\!}--- whether based on the Curry--Howard isomorphism or illative combinatory logic ---
fixed point combinators appear unexpectedly as the untyped skeletons of paradoxes, heralding inconsistency of the logic lying over the computational calculus.
\cite{B1993}
\cite{Howe1987}
\cite{CoquandHerbelin1994}
\cite{Geuvers07}
\cite{Curry1942}
\cite{MeyerReinhold1986}

It is an elementary fact that a term $Y$ is a fixed point combinator if and only if $Y$ is itself a fixed point of the combinator $\delta = \lam y x. x (y x)$.
\cite[6.5.3]{B84}
This can even be taken as the definition of fpcs: $Y \in \Lam$ is fpc iff $Y = \delta Y$.  Corrado B\"ohm noticed that also $Y \delta$ is fpc whenever $Y$ is. \cite[6.5.4]{B84}
For example, if $Y = \yc$ is Curry's fpc, then $Y \delta = \Theta$ is Turing's fpc.  A major open problem in the Lambda Calculus asks whether there exists a ``double fpc'' $Y$ satisfying $\delta Y= Y = Y \delta$.
Statman \cite{Statman1993} conjectures that no such $Y$ exists.
An early attack on this problem was undertaken by Intrigila \cite{Intrigila97}.  Unfortunately, Endrullis discovered a gap in the argument which seems difficult to overcome. As of this writing, the conjecture remains open.
For recent developments, see \cite{EHKP}, \cite{MPSS}.
We will also discuss the conjecture in Section \ref{s:4}.

B\"ohm's observations revealed that fpcs themselves have a compositional structure, where one constructs new fpcs from old by applying them to $\delta$.  Since then, other ``fpc generating schemes'' have been discovered and investigated by several authors. \cite{Scott75} \cite{EHKP} These contributions have confirmed that fpcs have a  rich mathematical structure indeed.

In this paper, we will explore such fpc generators ``in the abstract'', studying their general properties and providing a basic taxonomy.  We formulate several new problems, including a significant strengthening of Statman's conjecture.


\section{Notations and definitions}

\begin{nota}
We assume the reader is familiar with the basic notions of lambda calculus:
$\lambda$-terms, free variables, substitution, and beta-conversion.
We refer to \cite{B84} for background on these matters.
Here we shall employ the following symbols and notions.

\begin{itemize}
	\item $\Lambda$ is the set of $\lambda$-terms.
    $\Lambda^0 = \setof{M \in \Lambda \mid \fv(M)=\emptyset}$ is the set
    of \emph{closed} $\lam$-terms.
	\item $\fv(M)$ is the set of free variables of $M \in \Lambda$.
	\item $M[x:=N]$ is the result of capture-avoiding substitution of $N$ for $x$ in $M$.
	\item If $\vec N = (N_1,\dots,N_k)$ is a sequence of $\lambda$-terms, then $M \vec N = MN_1\cdots{}N_k$.
	\item $F^k(z) := F(F(\cdots F(z)\cdots))$, with $k$ $F$s.
	\item $\ic = \lambda x.x$, $\kc=\lambda xy.x$, $\cnc_k = \lambda x y. x^k(y)$,
	$\delta = \lambda y x. x (y x)$. 
	\item $M=N$ denotes beta conversion between $M$ and $N$.
	\item $M \thra N$ denotes beta reduction from $M$ to $N$.
	\item $M$ is \emph{solvable}, if $M \vec N = \ic$ for some $\vec N$.  Otherwise, $M$ is \emph{unsolvable}.
	\item $M \bte N$ if $M$ and $N$ have the same B\"ohm tree.\\  We note without proof that this relation can be defined using one axiom and one inference rule, the latter to be understood \emph{coinductively}
	(see \cite{KKSV97},\cite{EP11}, and especially \cite[Def.\,5.6]{Czajka2020}):
	\[
		\AXC{$M, N \! \phantom{\vec M} \! \text{unsolvable}$}
		\UIC{$M \bte N$}
		\DisplayProof
		\qquad \! \!
		\AXC{$M = \lambda \vec x. y \vec M \quad N = \lambda \vec x. y \vec N$}
		\AXC{$M_1 \bte N_1\quad \cdots \quad M_k \bte N_k$}
		\doubleLine
		\BIC{$M \bte N$}
		\DisplayProof
	\]
	\item $z\# M$ means $z \notin \fv(M)$.  For $S \subseteq \Lambda$,
	$z\#S$ means $z\#M$ for each $M \in S$.
    \item $z \notin M$ if there exists $N = M$ such that $z\#N$.
    $z \in M$ if $z \in \fv(N)$ for all $N = M$.
    \item $z \notinfty M$ if there exists $N =^\infty M$ such that $z\#N$.
    Otherwise, $z \in^\infty M$.
\end{itemize}
\end{nota}

\begin{defi}
	$Y \in \Lambda$ is a \emph{fixed point combinator (fpc)} if $Yx = x(Yx)$ for $x \# Y$.
\end{defi}

\begin{defi}
	$Y \in \Lambda$ is a \emph{weak fixed point combinator (wfpc)} if $Yx \bte x(Yx)$ for $x \# Y$.
\end{defi}

Notice that every fpc is a wfpc.  See Examples \ref{ex1} for both types of terms.

All (w)fpcs have the same B\"ohm tree, so
$Y \in \Lam$ is wfpc iff $Y \bte Y_0$ for some fpc $Y_0$.

A wfpc $Y$ can equivalently be given by a sequence of terms $(Y_n)$ with $Y=Y_0$ and $Y_n x = x (Y_{n+1} x)$, with $x \# \{Y_n, Y_{n+1}\}$.
\cite[Prop.\,3.9]{MPSS}
If $Y$ is fpc, then $Y_n=Y_0$ for all $n$.


 \newcommand{\FPC}{\mathsf{FPC}}
 \newcommand{\FPCo}{\mathsf{FPC}^0}
 \newcommand{\WFPC}{\mathsf{WFPC}}
 \newcommand{\WFPCo}{\mathsf{WFPC}^0}
 \newcommand{\fpct}{\quad \text{ fpc }}
 \newcommand{\wfpct}{\quad \text{wfpc }}

\begin{nota}
	We write $\FPC$ ($\WFPC$) for the set of fpcs (weak fpcs).
\end{nota}


\begin{nota}
	Henceforth, we shall often write $\mathsf{(W)FPC}$ in a sentence that is meant to apply to both $\FPC$ and $\WFPC$.  Such a statement should always be read as a conjunction of two statements: one, in which parentheses are ignored together with their contents, and another, where parentheses are removed but their contents remain.
\end{nota}


\begin{defi}
A \emph{(weak) fpc generating vector}, or \emph{(w)fgv}, is a sequence of terms $\vec G$ satisfying
\[ Y \in \wFPC \then Y \vec G \in \wFPC. \]
\end{defi}

\begin{prop} TFAE: \label{tfae} \\
$(i)\quad\ \ \vec G$ is wfgv.\\
$(ii)\quad \ Y \in \FPC \then Y \vec G \in \WFPC$.\\
$(iii)\quad Y \vec G \in \WFPC$ for some $Y \in \FPC$.
\end{prop}
\begin{proof}
$(i) \RA (ii)$. Let $\vec G$ be wfgv, $Y$ be fpc.
Then $Y$ is wfpc, and $Y \vec G$ is wfpc.

$(ii) \RA (iii)$. Trivial.

$(iii) \RA (i)$.  Let $Z$ be wfpc.  Then $Z =^\infty Y$.
Also $Z \vec G =^\infty Y\vec G$ is wfpc.
\end{proof}

\begin{cor} Every fpc generator is wfpc generator.
\label{fpwpgen}
\end{cor}
\begin{proof}
Let $\vec G$ be fpc generator. Pick $Y \in \FPC$.
Then $Y \vec G$ is fpc, hence $\vec G$ is wfgv.
\end{proof}

\begin{prop}
\label{conds}
Consider the following conditions on $\vec G$.
\newcommand{\pue}{\hspace{-0.3cm}}
\begin{align*}
\begin{array}{l r l}
(i)\hspace{2cm} &Y\phantom{j} \fpct \pue&\then Y \vec G \fpct \\
(ii) \hspace{2cm}&Y\, \wfpct \pue&\then Y \vec G \wfpct \\
(iii) \hspace{2cm} &Y\phantom{j} \fpct \pue&\then Y \vec G \wfpct \\
(iv) \hspace{2cm} &Y\, \wfpct \pue&\then Y \vec G \fpct \\
\end{array}
\end{align*}
 The following relations are valid:
\[ (iv) \then (i) \then (ii) \iff (iii) \]
\end{prop}

\begin{proof}
These relations simply summarize the facts noted above.
\end{proof}

\section{Examples and first observations}

\begin{exas}\leavevmode \label{ex1}
\begin{itemize}
    \item
    \emph{Turing's fpc.} Let $\Theta x = V V x$, where $V = \lambda v x. x (v v x)$. Then $\Theta \in \FPC$.
	\item
	\emph{Parametrized Turing's fpc.}
	For $M \in \Lam$, let $\Theta_M x = V V M x$, where $V = \lambda v m x. x (v v m x)$.
    Then $\Theta_M \in \FPC$.
		(This example can be generalized to have multiple parameters.)
  \item Let $z$ be a variable. Put $\Psi_z = W_z W_z \ic$,
    where $W_z = \lambda w p x. x (w w (z p) x)$.
    Then $\Psi_z \in \WFPC\setminus\FPC$.
	\item A slight variant of the above will play a central role in the proof of our main result.
	Let $c$ be a variable.  Put $\Ups = \lambda x. V_x \ic V_x$, where
	$V_x = \lam p v. x (v (c p) v)$. Then $\Ups \in \WFPC\setminus\FPC$.

	A nice feature of this wfpc is
	that it has a very simple reduction graph.
\end{itemize}
\end{exas}

\begin{prop}
	$\Theta_{M} = \Theta_{N} \then M = N$.
	\label{thetaz}
\end{prop}
\begin{proof}
	This is manifest upon inspecting the reduction graph of $\Theta_z$ --- the set of reducts of $\Theta_z$.  For a precise proof, see \cite[Lemma 4.1]{MPSS}.
\end{proof}
\pagebreak

\begin{exas}\label{exs}\leavevmode
\begin{itemize}
	\item Let $\vec G = ()$, the empty vector.  Obviously,
    $Y \in \mathsf{(W)FPC} \then Y \vec G = Y \in \mathsf{(W)FPC}$.\\
    We call this generator \emph{trivial}.
    In subsequent sections, we will tacitly assume all generators to be non-trivial.
    \item Fix a (w)fpc $Y$, and let $\vec G = (\kc Y)$.  Then
    $(\kc Y)$ yields the same (w)fpc on every input:
    \[ Z=(Z_0,Z_1,\dots) \in \WFPC \then Z_0 (\kc Y) = \kc Y (Z_1 (\kc Y)) = Y \in \mathsf{(W)FPC}.\]
    We call such generators \emph{constant}.
    Their only interesting feature is the fixed point $Y=Y\vec G$.

    \item  Recall that $\delta y x = x (y x)$.  It is easy to verify the following: \label{e:delta}
    \begin{itemize}
    \item
    $\delta^k(z)x = x^k(z x)$.
    \item
    If $Y$ is fpc, then $Y=\delta Y = \delta^k(Y)$.
    \item
    If $Y=(Y_n)$ is wfpc, then $Y_0 = \delta^k(Y_k)$.
    \end{itemize}
    Let $\vec G = (\delta)$.  Then
    $Y \in \FPC \then Y \delta x = \delta (Y \delta) x = x (Y \delta x) \then Y \delta \in \FPC$.\\
    As noted in the introduction, it is open whether there exists $Y {\in}\, \mathsf{(W)FPC}$ such that $Y {=}\, Y \delta$.

    \item Let $\vec G = (\lambda y. \Theta_y)$.  Then
    \label{theta}
    \[ Y \in \FPC \then Y \vec G = Y (\lambda y. \Theta_y)
    = (\lambda y. \Theta_y) (Y (\lam y. \Theta_y))
    = \Theta_{Y \vec G} \in \FPC.\]
    Furthermore, there exists fpc $Y$ such that $Y = Y \vec G $.\\
    Indeed, take $Y = \Theta (\lam x. \Theta_{x (\lam y. \Theta_y)})
    = \Theta_{Y (\lam y. \Theta_y)}$.  Then $Y \in \FPC$, and
    \[ Y (\lam y. \Theta_y)
    	= (\lam y. \Theta_y) (Y (\lam y. \Theta_y))
    	= \Theta_{Y (\lam y. \Theta_y)} = Y. \]

		\item Yet another single-term fgv is given by
    $\vec G = (\lambda y x. x (y (\kc [y,x]) \ic))$, where $[P,Q]=\lambda z. zPQ$:
		\begin{align*}
			Y \in \FPC \then Y G_0 x = G_0 (Y G_0) x
			&= x (Y G_0 (\kc [Y G_0,x]) \ic)\\
			&= x (G_0 (Y G_0) (\kc [Y G_0,x]) \ic)\\
			&= x (\kc [Y G_0, x] (\cdots) \ic)\\
			&= x ([Y G_0, x] \ic) = x (\ic (Y G_0) x) = x (Y G_0 x)
		\end{align*}
    \item
    The set of (w)fgvs is closed under composition:
    if $\vec G$ and $\vec G'$ are fgvs, then
    \[ Y \in \FPC \then Y \vec G \in \FPC \then Y \vec G \vec G' \in \FPC. \]
    Thus, $(\delta,\lam y. \Theta_y)$ and $(\lam y. \Theta_y,\delta)$ are both fgvs.
    \item
    Many other examples of fpcs and fgvs can be found in \cite{EHKP} and \cite{MPSS}.
\end{itemize}
\end{exas}


\begin{defi}
A (w)fgv $\vec G$ is \emph{injective} if for all (w)fpcs $Y,Y'$, $Y \vec G = Y' \vec G$ implies $Y = Y'$.
\end{defi}

\begin{prop} \label{noninj}
No non-trivial (w)fgv is injective.
\end{prop}

\begin{proof}
 Suppose $\vec G = (G_0,\dots,G_n)$, for $n \ge 0$, is injective. Since $\Theta \vec G = G_0(\Theta G_0) G_1 \cdots G_n$ is (w)fpc, $G_0$ must be solvable. That is, $G_0 \vec P = \ic$ for some closed $\vec P$.
 Put $Y = \Theta_{x \vec P}$, $Y' = \Theta_{x \vec P \ic}$.
 By Proposition \ref{thetaz}, $Y \neq Y'$.  Yet
\begin{align*}
 Y \vec G = \Theta_{G_0 \vec P} \vec G
          = \Theta_{\ic} \vec G
          = \Theta_{G_0 \vec P \ic} \vec G
          = Y' \vec G.
\end{align*}
It follows that $\vec G$ is not injective.
\end{proof}

(Notice that in the above proof both $Y$ and $Y'$ are closed, so even restricting injectivity hypothesis to closed terms, no non-trivial wfpc generator is injective.)

\begin{cor}
	\label{triv}
	Suppose wfgv $\vec G$ fixes every fpc: $Y \vec G = Y$ for all fpc $Y$.
	Then $\vec G$ is trivial.
\end{cor}

An interesting consequence of these observations is that there is no uniform way to ``B\"ohm out'' an inner level of a wfpc.

\begin{prop}  For $m > 0$, it is not possible to find terms $(M_0,\dots,M_n)$ such that
\begin{equation}
	\label{Zm}
 Z = (Z_n)\text{ wfpc}  \quad \then \quad Z_0 \vec M = Z_{m}.
\end{equation}
\end{prop}

\begin{proof}
Suppose such $\vec M = (M_0,\dots,M_n)$ exists.
Then $\vec M$ is a wfgv.
For every fpc $Y$, we have $Y \vec M = Y$, so every fpc is fixed
by $\vec M$.
(In particular, $\vec M$ is a fgv.)
By Corollary \ref{triv}, $\vec M$ is trivial: $\vec M = ()$.
But then $\vec M$ fixes every wfpc as well, and thus cannot satisfy the hypothesis in \eqref{Zm}.
\end{proof}

\section{The four main classes of generators} \label{s:4}

Let us consider again Statman's conjecture on the non-existence of fixed points of the generator $\vec G = (\delta)$.  This conjecture is intuitively compelling, because applying any fpc $Y$ to $\delta$ leads to a slowdown of head reductions that seems impossible to remove:
\begin{alignat*}{3}
	Y x &\thra_w x (Y') &&=_\beta x (Yx) \\
	Y \delta x &\thra_w \delta (Y'[x:=\delta]) x
	\thra_w x(Y'[x:=\delta] x) &&=_\beta x (Y \delta x)
\end{alignat*}
Upon closer inspection, the central property of the generator $(\delta)$ that this reasoning depends on is that $\delta(Y\delta)x$ \emph{adds to} the reduction length needed for the B\"ohm tree to develop, while still using the given fpc $Y$ in constructing this B\"ohm tree infinitely often.
Since any conversion between the two will synchronize their B\"ohm reductions, no such conversion can be possible.

If this reasoning proves to be correct for $\delta$, it should remain valid for any generator possessing the same property.  This leads us to the following definition and conjecture.

Recall, for any wfpc $(Y_n)$, $Y_0 = \lambda x. x^k(Y_k x) = \delta^k(Y_k) = \delta^k(z)[z:=Y_k]$.

\begin{defi}
	A generator $\vec G$ is \emph{accretive} if, for each $k$, $z \in^\infty \delta^k(z)\vec G$.
\end{defi}

That is, $\vec G$ is accretive if it actually uses every level of the input fpc
in constructing the output fpc, so that replacing $Y$ with any approximant will
also cut the B\"ohm tree of $Y \vec G$.

\begin{conj}
\label{conj}
If $\vec G$ is accretive, then there exists no $Y\in\WFPC$ such that $Y = Y \vec G$.
\end{conj}

\begin{rem}
	Conjecture \ref{conj} generalizes Statman's conjecture.
	Indeed,
	\[ \delta^k(z) \delta
	= (\lambda x. x^k(zx))\delta
	= \delta^k(z\delta) = \lambda x. x^k (z \delta x) \]
	Thus, $z \in^\infty \delta^k(z) \delta$ for all $k$, and the fgv $\vec G = (\delta)$ is accretive.
\end{rem}

Conjecture \ref{conj} is as sharp as possible: later in this section, we will show that every $\vec G$ which is \emph{not} accretive possesses a fixed point among the wfpcs.

The non-accretive generators can be naturally divided into several classes given below.  (We eschew the (W)FPC notation in the next definition to emphasize that there are actually four properties of generators that are being defined.)

From the earlier equality $Y_0 = \delta^k(Y_k)$, note that
for each (w)fpc $Y$, $k \ge 0$, we can write
\[Y \vec G  = \delta^k(Y')\vec G= G_0^k(Y' G_0) G_1 \cdots G_n.
\]
\newpage

\begin{defi} Let $\vec G$ be a wfgv.  Fix $z \# \vec G$.
\begin{itemize}
\item $\vec G$ is \emph{constant} if there is a $k$ such that $\,z \notin G_0^k(z) G_1 \cdots G_n$.
\item $\vec G$ is \emph{weakly constant} if there is a $k$ s.t.\ $z \notinfty G_0^k(z) G_1 \cdots G_n$.
\item $\vec G$ is \emph{compact} if there is a $k$ such that $\;G_0^k(z) G_1 \cdots G_n \in \FPC$.
\item $\vec G$ is \emph{weakly compact} if there is a $k$ s.t.\ $G_0^k(z) G_1 \cdots G_n \in \WFPC$.
\end{itemize}

The least $k$ satisfying one of these conditions is then called the \emph{modulus of constancy}, or \emph{modulus of compactness}, accordingly.  Note that $\vec G$ is accretive iff $\vec G$
is not weakly compact.
\end{defi}

From now on, let $\vec G$ be a possibly weak fgv.  We will omit freshness conditions $x \# Y$, $z \# \vec G$ etc., as they will always be obvious from the context.
\begin{prop}
Let $\vec G$ be a constant (w)fgv.  There is a term $Z$ such that
\[ Y \vec G = Z\]
for all wfpc $Y$.
Hence $Z$ is (w)fpc.
\end{prop}
\begin{proof}
Let $\vec G$ be constant, and let $k$ be such that
$z \notin G_0^k(z) G_1 \cdots G_n$.

That is, $z \notin \fv(Z)$ for some $Z \in \Lambda$
convertible to $G_0^k(z) G_1 \cdots G_n$.

Then for any wfpc $Y = (Y_0,Y_1,\dots)$, we have
\begin{align*}
Y \vec G = Y_0 \vec G &= Y_0 G_0 \cdots G_n\\
&= G_0 (Y_1 G_0) G_1 \cdots G_n\\
&= \qquad \vdots\\
&= G_0^k(Y_k G_0) G_1 \cdots G_n\\
&= G_0^k(z) G_1 \cdots G_n[z:=Y_k G_0]\\
&= Z [z:=Y_k G_0]\\
&= Z \hspace{3.5cm} \qedhere
\end{align*}
\end{proof}

\begin{prop}
The following observations are immediate.
\begin{enumerate}
	\item Every constant fgv is compact.
	\item Every constant (w)fgv is weakly constant.
	\item Every compact (w)fgv is weakly compact.
\end{enumerate}
\end{prop}

\begin{prop}  
	$\vec G$ is weakly constant iff $\vec G$ is weakly compact.
\end{prop}

\begin{proof}  Let $\vec G$ be a wfgv.  Then
	\begin{align*}
		z \notinfty G_0^k(z) G_1 \cdots G_n
		&\then G_0^k(z) G_1 \cdots G_n =^\infty
		G_0^k(\Theta G_0) G_1 \cdots G_n =
		\Theta \vec G \in \WFPC,\\
		G_0^k(z) G_1 \cdots G_n \in \WFPC
		&\then G_0^k(z) G_1 \cdots G_n =^\infty \Theta
		\then z \notinfty G_0^k(z) G_1 \cdots G_n. \qedhere
	\end{align*}
\end{proof}

\begin{prop}  Every (weakly) compact generator has a fixed point:
\label{fix}
\begin{enumerate}
	\item If $\vec G$ is compact fgv, there exists fpc $Y$ with $Y \vec G = Y$. \label{fix1}
	\item If $\vec G$ is weakly compact wfgv, there exists wfpc $Y$ with $Y \vec G = Y$. \label{fix2}
\end{enumerate}
\end{prop}

\begin{proof}
The construction is the same for both claims.  We will first treat the weak case, and then specialize the proof to the first claim as well.

Let $\vec G$ be a weakly compact wfgv.
Let $k$ be the modulus of weak compactness,
so that $G^k_0(z)G_1 \cdots G_n \in \WFPC$.
Put $F_0[z] := G^k_0(z) G_1 \cdots G_n$.  Since $F_0[z]$ is wfpc, write
\begin{equation}
 F_0[z] x = x(F_1[z] x) =x^2(F_2[z]x) = \cdots = x^k(F_k[z]x) = \cdots
 \label{e:fk}
\end{equation}
Define $Y:=\Theta (\lambda y. F_k[y G_0]) = F_k [Y G_0]$, and $X:=F_0[Y G_0]$.
Now compute
\begin{align*}
	X \vec G
	&= XG_0 \cdots G_n\\
	&= F_0[Y G_0] G_0 \cdots G_n
	& \text{by definition of $X$}\\
	&= G_0^k(F_k[Y G_0] G_0) G_1 \cdots G_n
    & \text{by \eqref{e:fk} with\ \ } [x := G_0, z:= Y G_0]\\
    &= G_0^k(Y G_0) G_1 \cdots G_n
    & \text{by definition of $Y$}\\
    &= F_0[Y G_0]
    & \text{by definition of $F_0$}\!\\
    &= X
    & \text{by definition of $X$}
\end{align*}
Since $F_0[z]$ is wfpc, so is $X = F_0[Y G_0]$, proving the second claim.

Now suppose that $\vec G$ was actually compact fgv.  Then $F_0[z]$ would be fpc, while all of the steps above would remain valid, with $F_0[z] = F_k[z]$ for each $k$.
Then $X = F_0[Y G_0]$ would be fpc as well, proving the first claim.
\end{proof}

\begin{rem}
The reader will recognize the converse to Proposition \ref{fix}(\ref{fix2}) as the contrapositive of Conjecture \ref{conj}.
 The converse to Proposition \ref{fix}(\ref{fix1}) seems plausible, but we do
 not have sufficient evidence to assert it as a formal conjecture.
\end{rem}
%
%

\section{Rectifying generators}

\begin{defi}
	A vector $\vec G$ is \emph{rectifying} if it satisfies condition $(iv)$ of Proposition \ref{conds}:
	\[ Y \in \WFPC \then Y \vec G \in \FPC \]
\end{defi}

\begin{exa}
	$\vec G = (\lambda y. \Theta_y)$ is rectifying:
	\[ (Y_n) \in \mathsf{WFPC} \then Y_0 (\lambda y. \Theta_y)
		= (\lambda y. \Theta_y) (Y_1 (\lambda y. \Theta_y))
		= \Theta_{Y_1 (\lambda y. \Theta_y)} \in \mathsf{FPC} \]
\end{exa}

In Example \ref{exs}, we saw that $(\lambda y. \Theta_y)$ has a fpc fixed point.  We shall presently see that so does every rectifying fgv.

Our original proof of this fact first showed that if $\vec G$ is rectifying, then $\vec G$ is weakly constant, and thus has a wfpc fixed point $Y$.  But then $Y = Y \vec G \in \FPC$ because $\vec G$ is rectifying, hence $Y$ is fpc.

Considering that compactness provides another sufficient condition for existence of fpc fixed points, it was natural to wonder whether rectifying and compact fgvs are related.  This led us to the following result.

\renewcommand{\then}{\quad \Longrightarrow \quad}
\begin{thm} \label{thm}
	A fgv $\vec G$ is compact iff it is rectifying.
\end{thm}

\begin{proof}
  \begin{description}
    \item[$(\RA)$]
	Suppose $G^k_0(z)G_1 \cdots G_n \in \FPC$.
	Then \begin{align*}
	     	&Y = (Y_0,Y_1,\dots) \in \WFPC \\
      \then &Y \vec G = Y_0 G_0 \cdots G_n
            = G_0^k(Y_k G_0) G_1 \cdots G_n
            = G_0^k(z) G_1 \cdots G_n [z:=Y_k G_0] \in \FPC.
	     \end{align*}
       \item[$(\LA)$]
    The intuition for this direction is that, although the B\"ohm tree of a wfpc $Y$ is infinite, only a finite part of it can be used in any conversion $\rho : Y \vec G x = x (Y \vec G x)$.  Thus, writing $Y x = x^k (Y_k x) = \delta^k(Y_k)x$ for large enough $k$ will ensure that $Y_k$ is not touched by any redex contractions.  Then the whole conversion $\rho$ could be lifted
    to $\rho = \sigma[z := Y_k]$, where \[\sigma : \delta^k(z) \vec G x = x (\delta^k(z) \vec G x).\]
    To formalize this intuition, suppose $\vec G$ is rectifying.
    Fix $c \# \vec G$.  Recall the wfpc $\Ups$ from Examples \ref{exs}:
    \begin{align*}
    	W_{x,p} &=\lam v. x (v (c p) v)\\
    	V_x &= \lam p. W_{x,p}\\
    	\Ups &= \lam x. V_x \ic V_x
    \end{align*}
    That is, $V_x = \lam p v. x (v (c p) v)$.  Note that $W_{x,p}$ and $V_x$ are normal forms.
    Let $\Ups^k_x = V_x c^k(\ic) V_x$.  The term $\Ups x$ reduces as follows:
    \begin{align*}
    	\Ups x \to \Ups^0_x \equiv V_x \ic V_x
    	&\to W_{x,\ic} V_x
    	\to x (V_x (c \ic) V_x)
    	\equiv x (\Ups^1_x)\\
    	&\to x(W_{x,c\ic} V_x)
    	\to x (x (V_x c^2(\ic) V_x))
    	\equiv x^2 (\Ups^2_x)\\
    	&\to \cdots\\
    	&\to x^k (\Ups^k_x)\\
    	&\to \cdots
    \end{align*}
    Since each term appearing in the above reduction sequence has a unique redex, the reduction is completely deterministic.
    That is --- the above sequence actually comprises the entire reduction graph of $\Ups^0_x$.
    The sequence also shows that $\Ups$ is a wfpc.
    It is not a fpc however, since $\Ups^0_x$ obviously has no reducts in common with $\Ups^1_x$.
    But $\vec G$ is rectifying, so $\Ups \vec G$ is fpc.
    By the Church--Rosser theorem, let $X$ be a common reduct
    \begin{equation}
 \Ups \vec G x \thra X \thla x (\Ups \vec G x).
 \label{CR}
    \end{equation}
    We will use these reductions to show that $\delta^k(z) \vec G \in \FPC$ for large enough $k$.

		The main idea behind the construction is as follows.
		Any finite reduction $\Ups M \thra X$ can be continued until all the descendants of $\Ups$ project the same number of steps from $\Ups^0_M$ in the sequence above.  Afterwards, all descendants of $M$ can be further synchronized by confluence.

		For example, if $M = \lambda y. [y \kc , y]$, then
		\begin{align*}
			\Ups (\ic M) \thra (\ic M) (\Ups^1_{\ic M})
			&\thra (\lambda y. [y \kc, y]) (V_{\ic M} c^1(\ic) V_{\ic M})\\
			&\thra [V_{\ic M} c^1(\ic) V_{\ic M} \kc ,
			V_{\ic M} c^1(\ic) V_{\ic M} ]\\
			&\thra [W_{\ic M, c^1(\ic)} V_{\ic M} \kc,
			\ic M (V_{\ic M} c^2(\ic) V_M)]
		\end{align*}
		and this reduction can be further continued to
		\begin{align*} [W_{\ic M, c^1(\ic)} V_{\ic M} \kc,
		\ic M (V_{\ic M} c^2(\ic) V_M)]
		&\thra [M(V_{M} c^2(\ic) V_M) \kc , M(V_{M} c^2(\ic) V_M )]\\
		&\; \equiv\; [M(\Ups^2_M),M(\Ups^2_M)].
	\end{align*}
    We proceed with the following sequence of claims, which are hopefully sufficiently clear not to warrant additional elaboration.
    \begin{enumerate}
        \item
				If $\Ups M \thra X$, then $X \thra X' \equiv C[V_{M_1}c^{k_1}(\ic)V_{M'_1},\cdots,V_{M_k}c^{k_m}(\ic)V_{M'_m}]$, with $M \thra M_i$, $M \thra M'_i$, and every occurrence of $c$ in $X'$ being displayed in the subterm $c^{k_i}(\ic)$ in one of the holes in $C[]$.
				\item If $\Ups M \thra X$, then $X \thra X' \equiv C[\Ups^{k_1}_{M_1},\dots,\Ups^{k_m}_{M_m}]$, with $M \thra M_i$ and every occurrence of $c$ being uniquely determined by its occurrence in some $\Ups^{k_i}_{M_i}$.
        This is obtained from above by finding a common reduct for each $M_i$,$M_i'$.
				\label{p2}
    	\item If $C[\Ups M] \thra X$, then $X \thra C'[\Ups^{k_1}_{M_1},\cdots,\Ups^{k_m}_{M_m}]$, with the same conditions on $M_i$ and occurrences of $c$ as before.
			This is obtained from the previous point by factoring the reduction into a part that does not depend on $\Ups M$,
			$C[x] \thra D[x \vec P_1,\dots, x \vec P_k]$, followed by reductions inside $D[\cdots]$ which are treated separately via \ref{p2}.  (See ``Barendregt's Lemma'' in \cite[Exercise 15.4.8]{B84})
    	\item If $C[\Ups M] \thra X$, then $X \thra C'[\Ups^k_{N}, \cdots, \Ups^k_N]$, where $M \thra N$ and each occurrence of $c$ being uniquely determined by its occurrence in some $\Ups^k_N$.
    	This is obtained from the previous claim by ``bumping all $\Ups^{k_i}$s along'' to stage $k \ge \max \setof{k_i}$, and letting $N$ be a common reduct of all the $M_i$s.
    	\item If the reduction $\rho : C[\Ups M] \thra C'[\Ups^k_N,\cdots,\Ups^k_N]$ is obtained by the algorithm given in the previous steps, then $\rho$ lifts to the instance
    	$\rho = \sigma[ux:=\Ups^k_x]$, where
    	\[ \sigma : C[\delta^k(u) M] \thra C'[u N,\cdots,u N]. \]
    \end{enumerate}
    And now we are done!  The common reductions in \eqref{CR} can be both continued to
    \[ \Ups \vec G x \thra C'[\Ups^k_N,\cdots,\Ups^k_N]
    \thla x (\Ups \vec G x)\]
    so that all of the descendants of $\Ups$ (under \emph{both} reductions) are displayed in the context.  (This follows from the fact that every occurrence of $c$ is witnessed in some $\Ups^k_N$, and $c$ was chosen to be fresh.  The variable $c$ acts as a ``label'' for the unfolding depth of $\Ups$.)

    The conclusion of the last step therefore holds for both of these reductions, and so
    \[ C[\delta^k(u)M] := [\delta^k(u)G_0]G_1\cdots G_n
    	\thra C'[uN,\dots,uN] \thla x(\delta^k(u)G_0\cdots G_n).
    \]
    That is, $G_0^k(uG_0)G_1\cdots G_n \in \FPC$.
		Applying the substitution $[u:=\kc u]$, we find that
		 $G_0^k(u)G_1\cdots G_n \in \FPC$ as well.\qedhere
   \end{description}
 \end{proof}

\begin{cor}
	Every rectifying fgv has a fixed point in $\FPC$.
\end{cor}

\begin{rem}
	The proof of the nontrivial direction of Theorem \ref{thm} suggests a deeper connection between uniform properties (finite conversions) and terms obeying a coinductive pattern (such as wfpcs).  In the next section, we will see a different application of the same type of argument.  There seems to be a more general ``continuity principle'' at work here that could be worthwhile to isolate.
\end{rem}

We finish this section with an example of a weakly constant fgv which is not rectifying.  It follows that, even restricting to fgvs, compactness is indeed stronger than weak compactness.

This gives the full picture of the relationships between various classes of wfgvs we defined here.  These relationships are summarized in Figure \ref{fig}.

\begin{prop}
There exist weakly constant
fpc generators which are not rectifying.
\end{prop}

\begin{proof}
Consider the following combinators:
\begin{align*}
P x y &=y x\\
Q y z &=z (y Q z)\\
W w p z &=z (w w (z p) z) \\
R y z &=W W (y Q z) z
\end{align*}

First we observe that $(P,Q)$ is an fgv:
for $Y$ fpc, we have
\begin{equation}
\label{ypqx}
Y P Q x
=P (Y P) Q x
= Q (Y P) x
= x(Y P Q x).
\end{equation}

We claim that $(P,Q)$ is not rectifying.
If it was, then by the previous theorem, it would be compact, hence weakly compact.
To the contrary, \eqref{ypqx} shows that $(P,Q)$ is accretive.  So it cannot be rectifying.

Next, we verify that $(P,R)$ is again fgv:
\begin{align*}
Y P R x
= P (Y P) R x
= R (Y P) x
&= W W (Y P\, Q\, x)\, x\\
&= x (W W (x (Y P\, Q\, x))\, x)\\
&=_{\eqref{ypqx}} x (W W (Y P\, Q\, x)\, x)
= x (Y P R\ \! x)
\end{align*}
At the same time, we claim $(P,R)$ is weakly constant with modulus $1$:
\begin{align*}
 P^1(z) R x = P z R x
= R z x
&= W W (z Q x) x\\
&= x (W W (\cdots) x)\\
&= x^2 (W W (\cdots) x)\\
&= \cdots\\
&= x^n (\cdots)
\end{align*}
The variable $z$ is being pushed to infinity, and does not appear
in the B\"ohm tree of $P z R x$---nor
in the B\"ohm tree of $P z R = \lambda x. P z R x$.
That is, $z \notin^\infty P^1(z)R$.
Indeed, $\vec G = (P,R)$ is weakly constant.  We claim it is not rectifying.

For a wfpc $Z$, the term $Z P R$ reduces as follows:
\begin{align*}
  Z P R x
\thra P (Z P) R x
\to^2 R (Z P) x
&\to^2 W W (Z P Q x) x\\
&\to^3 x (W W (x (Z P Q x)) x)\\
&\to^3 x^2 (W W (x^2 (Z P Q x)) x)\\
&\to \qquad \vdots\\
  Z P R x \quad
\thra^{z_0 + 2 + 2 + 3n}\hspace{-0.5cm} &\quad \ \ x^n(W W (x^n (Z P Q x)) x)
\end{align*}
From this analysis, it is manifest that any common reduction
\[Z P R x\ \ \thra\ \ \cdot \ \ \thla\ \ x(Z P R x) \]
must contain a common reduction between
\[x^n(Z P Q x)\ \ \thra\ \ \cdot \ \ \thla\ \ x^{n+1}(Z P Q x). \]
As we observed earlier, $(P,Q)$ is not rectifying, so there
exist wfpcs $Z$ for which such conversion is not possible.
Thus $(P,R)$ is not rectifying either.  This completes the proof.

(Note that the modulus of constancy can be adjusted to any $k > 0$ by passing the argument of the generator into the head position $k$ times before pushing it to infinity.)
\end{proof}

\begin{figure}
	\centering
	\includegraphics[scale=0.5]{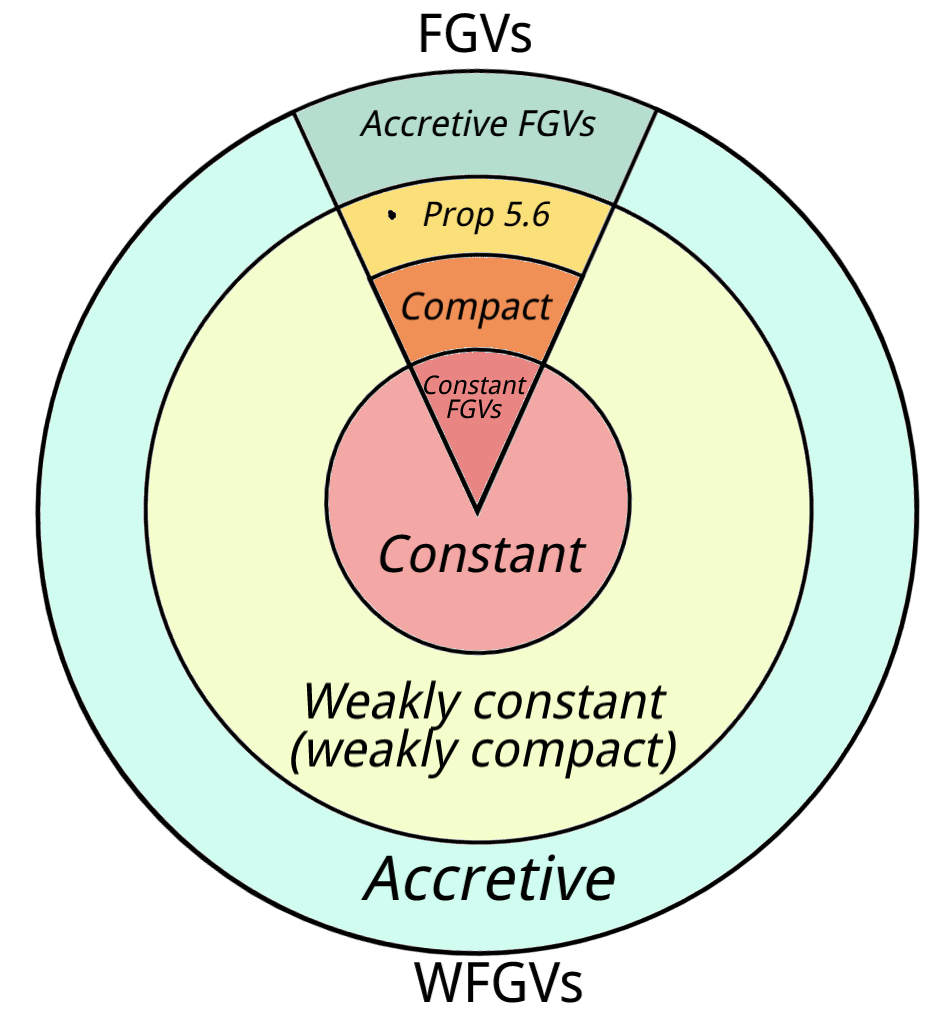}
	\caption{The hierarchy of (weak) fpc generators}
	\label{fig}
\end{figure}

\section{The monoid of wfgvs}
\newcommand{\gmon}{\mathcal G}
\newcommand{\fmon}{\mathcal F}
\newcommand{\gcat}{\odot}

The wfpc and fpc generators have an obvious monoid structure:
\[(G_0,\dots,G_n) \gcat (G_0',\dots,G_m')
\df (G_0,\dots,G_n,G_0',\dots,G_m') \]
The identity is the trivial generator $()$.
The concatenation operation is associative, and satisfies
the identity laws.
We thus have a monoid $(\gmon,\gcat,())$ of wfgvs,
containing a submonoid $(\fmon,\gcat,())$ of fgvs.

Neither of these monoids is finitely generated, as there are infinitely many constant fgvs of the form $\kc \Theta_M$ that cannot be obtained by composition of more elementary ones.

\subsection*{Extensional equality}

Since the primary interest in (w)fgvs is in their ability to generate new (w)fpcs from old, it is natural to identify generators having the same functional behavior.

\begin{defi}
	We say a fgv or wfgv $\vec G$ is \emph{extensionally equal} to $\vec G'$, written $G \simeq G'$, if for every \emph{fpc} $Y$, $Y\vec G = Y \vec G'$.
  (Note that this is an equivalence relation on $\gmon$, preserved by $\gcat$.)
\end{defi}

\newcommand{\Cc}{\mathtt{C}}
\begin{exas}
\begin{itemize}
	\item If $\vec G$ is a constant generator, say, $Y \vec G = Z$ for all $Y$, then $\vec G \simeq (\kc Z)$:
	\[ Y (\kc Z) = \kc Z (Y' (\kc Z)) = Z = Y \vec G \]
	\item
    Recall the combinator $\Cc = \lambda f x y. f y x$.
    Note that $\Cc \kc = \kc \ic$, and $\Cc (\Cc \kc) = \Cc (\kc \ic) = \kc$.
  Let $G y z = z (y (\Cc z)) (\delta (y (\Cc z)))$.
  Then $(G,\kc)$, and $(G, \Cc \kc)$ are fgvs, and $(G,\kc) \simeq (G,\Cc \kc)$:
  \begin{align*}
  	Y \in \FPC \then
  	Y G \kc &= G (Y G) \kc
  	= \kc (Y G (\Cc \kc)) (\delta (Y G (\Cc \kc)))
  	= Y G (\Cc \kc) \\
  	&= G (Y G) (\Cc \kc)
  	= (\Cc \kc) (Y G (\Cc (\Cc \kc))) (\delta (Y G (\Cc (\Cc \kc))))\\
  	& \hspace{2.415cm} = (\kc \ic) (Y G \kc) (\delta (Y G \kc)) = \delta (Y G \kc) \in \FPC
  \end{align*}
\end{itemize}
\end{exas}

The reason that in the definition of $\simeq$ the quantifier ranges over fpcs both in the case of fgvs as well as wfgvs is that, when the quantifier is taken over all wfpcs, it makes the resulting notion of equality much more restrictive, as we shall now demonstrate.

Because we obviously want equal fgvs to remain equal as wfgvs, the definition of extensional equality is expressed in terms of behavior on fpcs in both contexts.

\begin{prop}
	\label{ee}
	If $Y \vec G = Y \vec G'$ for every wfpc $Y$,
	then for some $k$, $\delta^k(z)\vec G = \delta^k(z)\vec G'$.
\end{prop}

\begin{proof}
	This statement follows by the same reasoning as used in Theorem \ref{thm}.
	Take $z \# \vec G, \vec G'$, and let $\Ups=\Ups_z$ be the canonical wfpc defined there with a deterministic reduction graph that uses the variable $z$ to track its unfolding history.
	The argument subsequently showed how every conversion
	$C[\Ups M] \thra X \thla C'[\Ups M]$ can be extended through $X \thra X'$, such that
	$X' = D[\Ups^k_N,\cdots,\Ups^k_N]$, with $M \thra N$ and every occurrence of $z$ in $X'$ to be found among the displayed $\Ups^k_N$.  We could then conclude that the common reduction may be lifted to a finite truncation of $\Ups$.
	In the present case, our starting conversion has the form
	\begin{equation}
	\label{Ups}
    C[\Ups G_0] := [\Ups G_0] G_1 \cdots G_n = [\Ups G_0'] G_1' \cdots G_{n'}' =: C'[\Ups G_0']
	\end{equation}

	To justify application of the same argument, we should thus argue why $G_0 = G_0'$.
	Let $X$ be a reduct of $C[\Ups G_0]$.  By recalling the reduction graph of $\Ups$, it is evident that every innermost occurrence of $z$ in $X$ is applied to a reduct of $G_0$.
	If $X$ is also a reduct of $C'[\Ups G_0']$, then the same conclusion will hold, with $G_0'$ in place of $G_0$.   Thus, the very fact of occurrence of $z$ in $X$ forces $G_0$ and $G_0'$ to be convertible.
	Of course, if $z$ does not occur in $X$ at all, that only means that all descendants of $\Ups$ have already been erased, in which case there is nothing left to prove.
	So $G_0 = G_0'$.  We can thus adjust conversion in \eqref{Ups} to
	\[ C[\Ups G_0] \equiv [\Ups G_0] G_1 \cdots G_n
		= [\Ups G_0] G_1' \cdots G_{n'}' \equiv C'[\Ups G_0] = C'[\Ups G_0'] \]
    where the conversion on the right takes place inside the subterm that $\Ups$ is applied to.

    Now we extend the other conversion to a common reduct
    \[ [\Ups G_0] G_1 \cdots G_n \thra D[\Ups^k_N,\cdots,\Ups^k_N]
    	\thla [\Ups G_0] G_1' \cdots G_{n'}' \]
    and proceed, as in the proof of Theorem \ref{thm}, to lift these reductions to
    \[ [\delta^k(u)G_0] G_1 \cdots G_n \thra D[uN,\dots,uN] \thla [\delta^k(u)G_0]G_1' \cdots G_{n'}'. \]
    Converting $G_0$ in the right term to $G_0'$, we obtain the desired result.
\end{proof}
From now on, we will consider the monoid $\gmon$ up to extensional equality.
We will also write concatenation of vectors by juxtaposition: $\vec F \vec G = \vec F \gcat \vec G$.

\subsection*{Ideals}
\begin{defi}
	A \emph{two-sided ideal} in a monoid $(M,\cdot)$ is a set $I \subseteq M$ such that
	\[ i \in I,\, m \in M \then i\cdot m,\, m \cdot i\, \in\, I. \]
\end{defi}

\begin{prop}\leavevmode
\label{ideal}
\begin{enumerate}
	\item
The constant generators form the minimal ideal in both monoids.
\item
The weakly constant/weakly compact wfgvs form a two-sided ideal in $\gmon$.
\item
	The compact fgvs form a two-sided ideal in  $\fmon$
	(and a right ideal in $\gmon$).
\end{enumerate}
\end{prop}

\begin{proof}
\begin{enumerate}
	\item Let $\vec G$ be constant, so that $Y \vec G = Z$ for all (w)fpc $Y$.
  Let $\vec G'$ be arbitrary.
  Then $Y \vec G \vec G' = Z \vec G'$ for all $Y$, and $Y \vec G' \vec G = Z$
  for all $Y$.
  Thus $\vec G \vec G'$ and $\vec G' \vec G$ are constant.
	Moreover, since any ideal includes the constant generators by composition on the right, these generators together constitute the minimal ideal.
	\item
  Let $\vec G \in \gmon$ be weakly constant, so that $z \notinfty G_0^k(z)G_1 \cdots
    G_n$.
    Let $\vec G' \in \gmon$ be arbitrary.  Clearly,
    $z \notinfty G_0^k(z) G_1 \cdots G_n G_0' \cdots G_m'$.
    That is, $\vec G \vec G'$ is weakly constant.
On the other hand, we know that $\vec G'$ maps wfpcs to themselves:
\[ (\lambda y. y\vec G') : \WFPC \to \WFPC\]
All wfpcs have the same B\"ohm tree, and in the tree topology, its neighborhood basis consists of the set $\setof{\lambda x. x^n(\Omega) \mid n \ge 0}$.
By Continuity Theorem \cite[14.3.22]{B84}, there exists
$l \ge 0$ such that
\begin{align*}
(\lambda x. x^l(z))\vec G' = (\lambda y.y \vec G') (\lambda x.x^{l}(z))
&= \lambda x.x^k(X)
\end{align*}
for some $X$, possibly containing $z$.
And yet,
\begin{align*}
  z \notin^\infty G_0^k(X)G_1 \cdots G_n
= (\lam x.x^k(X))\vec G
= (\lambda x. x^l(z)) \vec G' \vec G
= (G'_0)^l(z)G'_1\cdots{}G'_m\vec G.
\end{align*}
Indeed, $\vec G' \vec G$ is weakly constant.

\item
It is immediate that the rectifying fgvs form a two-sided ideal
in $\fmon$.  By Theorem \ref{thm}, so do the compact ones.
Also, for $\vec G, \vec G' \in \gmon$, $\vec G'$ rectifying
clealry implies $\vec G \vec G'$ rectifying. \qedhere
\end{enumerate}
\end{proof}

\subsection*{Green's relations}

\newcommand{\gLeft}{\mathpzc{L}}
\newcommand{\gRight}{\mathpzc{R}}
\newcommand{\gle}{\preccurlyeq}
The structure of many monoids can be characterized in terms of Green's relations.  Here we record several observations about these relations in $\gmon$, which could be useful for future study of this monoid.
\begin{defi}
	For $\vec G, \vec G' \in \gmon$, put
	\begin{align*}
    \gLeft(\vec G)  &\!\quad =\quad \setof{\vec H \vec G \mid \vec H \in \gmon}\\
    \gRight(\vec G) &\!\quad =\quad \setof{\vec G \vec H \mid \vec H \in \gmon}\\
    \vec G \gle_{\gLeft} \vec G' &\iff \,\gLeft(\vec G) \subseteq \gLeft(\vec G') \, \iff  \vec G \in \gLeft(\vec G')\\
    \vec G \gle_{\gRight} \vec G' &\iff \gRight(\vec G) \subseteq \gRight(\vec G')\iff \vec G \in \gRight(\vec G')\\
    \vec G \sim_{\gLeft} \vec G' &\iff \,\gLeft(\vec G) = \gLeft(\vec G')\\
    \vec G \sim_{\gRight} \vec G' &\iff \gRight(\vec G) = \gRight(\vec G').
	\end{align*}
\end{defi}
\newcommand{\vg}{\vec G}
\newcommand{\vf}{\vec F}

\begin{enumerate}
		\item If $\vec G \simeq (\kc Z)$ is a constant generator, then
		$\gLeft(\vec G) = \setof{\kc Z}$, so all constant generators are each in their own left class.
		That is, $\vec G \in \gLeft (\kc Z)$ implies $\vec G \simeq (\kc Z)$.
		\item On the other hand, $(\kc Z,\kc Z') \simeq (\kc Z')$,
		thus $\kc Z' \in \gRight (\kc Z)$.  Since the choice of $Z,Z'$ was arbitrary,
		$\kc Z \sim_{\gRight} \kc Z'$ for all $Z$ and $Z'$.
		That is, constant generators are all in the same right class.
		Since constant generators form an ideal, $\vec G {\sim_\gRight} (\kc Z)$ or $\vec G \gle_\gRight (\kc Z)$ imply $\vec G \simeq (\kc Z')$ for some $Z'$.
		So $\gRight(\kc Z) = \setof{(\kc Z') \mid Z' \in \mathsf{WFPC}}$
		is the ideal of all constant generators.
		\item Similarly, if $\vec G$ is weakly constant, then so is every element of $\gLeft(\vg)$ and $\gRight(\vg)$.
		That is, the only (w)fgvs that can be congruent to $\vec G$ modulo $\sim_\gLeft$ or $\sim_\gRight$ are again weakly constant.
		\item If $\vec G$ is compact, then so is every element of $\gLeft(\vg)$.
	 	When restricted to $\fmon$, both $\gLeft(\vg)$ and $\gRight(\vg)$ consist of compact generators.
		\item Suppose $\vg \sim_\gRight \vg'$.  Then we can find $\vf, \vf' \in \gmon$ such that $\vg \simeq \vg' \vf$, and $\vg' \simeq \vg \vf'$.
		But then $\vg \simeq \vg \vf' \vf$,
		and $\vg' \simeq \vg' \vf \vf'$.
		If $\vg \simeq \vg \vf' \vf$, then for every $Y$, $Y \vg = Y \vg \vf' \vf$ is a fixed point of $\vf' \vf$.
		By Conjecture \ref{conj}, $\vf' \vf$ is weakly constant.
		By Proposition \ref{ideal}, so is
		$\vg \vf' \vf$.
		But $\vg \vf' \vf \simeq \vg$.  So $\vg$ is weakly constant.
		Of course, everything we just said applies to $\vg'$ as well.
		We conclude that, modulo Conjecture \ref{conj}, nontrivial
		$\sim_\gRight$-relations can only exist between weakly constant wfgvs.
\end{enumerate}
In fact, we believe that for accretive $\vec G, \vec G'$, either $\vec G \sim_\gRight \vec G'$ or $\vec G \sim_\gLeft \vec G'$ implies $\vec G \simeq \vec G'$.
Indeed, this is a consequence of the ``freeness'' conjecture we formulate next.

\subsection*{Unique factorization of accretive generators}
\begin{defi}
	An accretive generator $\vec G$ is \emph{prime} if
	$\vec G = \vec G_1 \vec G_2$ implies $\setof{\vec G_1,\vec G_2} = \setof{\vec G,()}$.
\end{defi}
The following ``unique factorization conjecture'' states that the accretive generators are freely generated by the prime ones.
\begin{conj} \label{free}
If $\vec G \in \gmon$ is accretive, then $\vec G = \vec G_1 \cdots \vec G_k$, where each $\vec G_i$ is prime.

Furthermore, this decomposition is unique up to extensional equality.  That is, for all prime $\vec G'_1, \dots, \vec G'_{k'}$,
if $\vec G = \vec G'_1 \cdots \vec G'_{k'}$, then $k = k'$
and $\vec G_i \simeq \vec G'_i$ for all $i \in \setof{1,\dots,k}$.
\end{conj}

The conjecture implies that any relations in the monoid of wfpc generators under composition can only arise between non-accretive i.e., weakly compact generators.  In particular, the equations
\begin{align} \label{eq1}
	\vec F \vec G &= \vec F \\ \vec F \vec G &= \vec G \label{eq2}
\end{align}
admit no (non-trivial) solutions among the accretive generators.
As a result, the left and right classes of accretive generators would indeed all be distinct if Conjecture \ref{free} was to be validated.

And what about the solutions to \eqref{eq1} or \eqref{eq2} among the rest of $\gmon$?  The following examples show that there indeed exist solutions to these equations under certain conditions.  In all cases, (weak) compactness plays an essential role.
\begin{prop}
For $\vg$ weakly constant, there exists non-constant $\vf$ with $\vf \vg \simeq \vf$.
(In particular, $\vf \gle_\gLeft \vg$.)
\end{prop}

\begin{proof}
The idea is to make $\vec F$ generate the fixed points of $\vg$ according to the scheme in Proposition \ref{fix}.
Let $k$, $F_0$, $F_k$ be chosen as in the proof of that proposition.
Put $A = \lambda y b. b (y \delta)$, $B = \lambda y. F_0[y (\lambda u.F_k[uG_0]) G_0]$, and $\vec F = (A,B)$.
Observe that
\begin{align*}
	Y \in \FPC \then Y\! A \delta &= A (Y\! A) \delta = \delta (Y\! A \delta)\\
	Y\! A B &= A (Y\! A) B = B (Y\! A \delta) = F_0[ (Y\! A \delta)
	(\lam u. F_k[uG_0]) G_0].
\end{align*}
Since $Y\! A \delta$ is thereby forced to be fpc, it follows that $Y\! A B = F_0[U G_0]$, where $U = F_k[U G_0]$.  This allows us to calculate as in the proof of Proposition \ref{fix} that $Y\! A B$ is a fixed point of $\vec G$:
\[ Y \vec F \vec G = Y\! A B \vec G = Y\! A B\]
Note however, that $\vec F$ will not be constant in general, because it uses its fpc argument to define $U$.
\end{proof}

\begin{prop}
Let $\vec F=(F_0,\cdots,F_n)$ be wfgv with $n \ge 1$.
There exists a compact fgv $\vec G$ such that $\vec F \vec G \simeq \vec G$.
(In particular, $\vec G \gle_\gRight \vec F$.)
\end{prop}

\begin{proof}
		First, recall that $F_0 = \lambda v_0..v_l. v_m \vec P$ is solvable.  Since $Y \vec F = F_0 (Y' F_0) F_1 \cdots F_n$, we also know that the head variable $v_m$ cannot be $v_0$, for otherwise the result would be unsolvable, while it must be a wfpc.

		We let $\vec G = (F_0,G_1,\cdots,G_{n+1})$.  We will only need to specify a couple of $G_i$s.
		Set $G_m = \lambda \vec p. \lam g_{l+1},...,g_{n+1}. \Theta_{g_{n+1} (F_0 v_0 F_1 \cdots F_n)}$, $G_{n+1} y = \Theta(\lambda gy.g(y \vec F)) = G_{n+1} (y \vec F)$.
		\begin{align*}
			Y \vec G &= G_0 (Y' G_0) G_1 \cdots G_{n+1}\\
			&= F_0 (Y' F_0) G_1 \cdots G_{n+1}\\
			&= G_m \vec P [v_0:=Y' F_0] [v_i := G_i]_{1 \le i \le l} G_{l+1}
			\cdots G_{n+1}\\
			&= \Theta_{G_{n+1} (F_0 (Y' F_0) F_1 \cdots F_n)}\\
			&= \Theta_{G_{n+1} (Y F_0 \cdots F_n)} \\
			&= \Theta_{G_{n+1} (Y \vec F)}  \tag{$\star$}\\
			&= \Theta_{G_{n+1} (Y \vec F \vec F)} \\
			&= \Theta_{G_{n+1} ((Y \vec F) \vec F)}\\
			&= (Y \vec F) \vec G \quad \textrm{by ($\star$), with $Y:=Y\vec F$} \qedhere
		\end{align*}
	\end{proof}
	Our final observation is a corollary to one of the first ones.

\begin{prop}
	The monoid $\gmon$ is \emph{zerosum-free}: If $\vf\vg \simeq ()$, then $\vf\simeq()\simeq\vg$.
\end{prop}

\begin{proof}
	Suppose $\vf\vg \simeq ()$.  Then, considered as endofunctions on $\WFPC/{=_\beta}$, $\vg$ acts as a left inverse of $\vf$, making $\vf$ a split mono (modulo beta).
	But we have seen in Proposition \ref{noninj} that no wfgv is injective, so no wfgv can be monic.
	Specifically, take $Y \neq Y'$ such that $Y \vf = Y' \vf$.
	Since $() \simeq \vf \vg$,  we have $Y = Y \vf \vg = Y' \vf \vg = Y'$, a contradiction.
\end{proof}

\section{Concluding remarks}

In this paper,we have broached the topic of abstract fpc generators.  Our first investigations revealed that these operators naturally fall into a few robust classes.  We established elementary relationships between these classes.

What becomes clear from our investigations is that there is yet much to be uncovered about the structure of fixed point combinators.  Some of the possible future research directions include the following.

\begin{enumerate}
	\item The most pressing issue is the status of Conjecture \ref{conj}.  All the evidence available points to this conjecture being true, yet current techniques in untyped lambda calculus decidedly come up short in settling the question.  However it will be decided, the insights to be gathered from the new approaches will greatly deepen our understanding of lambda terms.
	\item Of course, one could take the next step and ask whether the converse to the first claim in Proposition \ref{fix} is also valid.  Considering how difficult the former question is, this one will likely remain out of reach for the foreseeable future.
	\item What else can be said about the structure of the monoid $\gmon$?  Is the submonoid of accretive generators freely generated by the prime generators, as Conjecture \ref{free} asserts?  What about the (weakly) compact generators?  How can their compositional structure be characterized?
	\item Since the monoid of (w)fgvs naturally acts on the set of (w)fpcs, how much of the structure of fpcs is captured by this monoidal action?  Does every fpc have a representation in terms of the prime elements of the monoid --- again, modulo extensional equality, and the ideal of compact generators?
	\item Finally, while not directly relevant to the earier discussion, an answer to the following question could also shed light on recursion-theoretic properties of FPCs:

	Let $\yc$ be Curry's simplest fpc.  Is $\setof{\#M \mid M = \yc}$ a decidable subset of $\FPC$?
	Specifically, does there exist a term $\Delta_{\yc}$ satisfying, for all \emph{closed
	$Y \in \FPC$}, the following:
	\begin{align*}
		\Delta_{\yc} \ulcorner Y \urcorner x y = \begin{cases}
		                       	x & \yc =_\beta Y\\
		                       	y & \yc \neq_\beta Y
		                       \end{cases}
	\end{align*}
	Notice that Scott's theorem does not apply here because $\FPC$ is not all of $\Lambda$, but is only a computably enumerable subset of it.  $\Delta_{\yc}$ is allowed to diverge outside of this set.
\end{enumerate}
  The recent paper \cite{MPSS} proposes another approach to Statman's conjecture based on simple types.  We note that the generalization of the conjecture stated there is consistent with ours, since every simply-typed generator is accretive thanks to strong normalization of ($\yc$-free) typed terms.

I would like to thank Jan Willem Klop, Joerg Endrullis, Dimitri Hendriks, Giulio Manzonetto, and Stefano Guerrini for wonderful discussions about fpc generators.

I would also like to extend my gratitude to the anonymous referees whose many
suggestions have significantly improved this paper.

\end{document}